\documentclass[a4paper]{article}

\usepackage{fullpage}
\usepackage{amsmath,amsthm,amssymb}
\usepackage{xspace}
\usepackage{graphicx}

\newtheorem{lemma}{Lemma}
\newtheorem{theorem}[lemma]{Theorem}
\newtheorem{fact}[lemma]{Fact}

\usepackage{amsmath,amssymb}
\newcommand{\FAM}{\mathfrak{F}}
\newcommand{\Ds}{\mathfrak{D}}
\newcommand{\R}{\mathbb{R}}
\newcommand{\Th}{\theta}
\newcommand{\Pb}{\bar{P}}
\newcommand{\Tb}{\bar{T}}

\newcommand{\Reals}{\mathbb{R}}
\newcommand{\m}{^{\circ}}
\newcommand{\strip}{\mathfrak{S}}
\newcommand{\eps}{\varepsilon}

\usepackage{color}

\usepackage{hyperref}

\newcommand{\dop}[2]{\langle #1, #2\rangle}

\newenvironment{denseitems}{\list{$\bullet$}%
  {\labelwidth3em\itemsep0pt\parsep0pt\topsep0.6ex}}{\endlist}

\let\geq\geqslant
\let\leq\leqslant
\makeatletter
\ifx\showkeyslabelformat\@undefined\else
\renewcommand{\showkeyslabelformat}[1]{\normalfont\tiny\ttfamily#1}
\fi
\partopsep\z@ \textfloatsep 10pt plus 1pt minus 4pt
\def\section{\@startsection {section}{1}{\z@}{-3.5ex plus -1ex minus
    -.2ex}{2.3ex plus .2ex}{\large\bf}}
\def\subsection{\@startsection{subsection}{2}{\z@}{-3.25ex plus -1ex
    minus -.2ex}{1.5ex plus .2ex}{\normalsize\bf}}
\def\@fnsymbol#1{\ensuremath{\ifcase#1\or *\or 1\or 2\or 3\or 4\or
    5\or 6\or 7 \or 8\ or 9 \or 10\or 11 \else\@ctrerr\fi}}
\makeatother

\title{A Generalization of the Convex Kakeya Problem%
  \thanks{H.-K.A.~was supported by NRF grant 2011-0030044 (SRC-GAIA)
    funded by the government of Korea.  J.G.~is the recipient of an
    Australian Research Council Future Fellowship (project number
    FT100100755).  O.C.~was supported in part by NRF grant
    2011-0030044 (SRC-GAIA), and in part by NRF grant~2011-0016434,
    both funded by the government of Korea.}}

\author{Hee-Kap Ahn\thanks{POSTECH, South Korea.
Email: {heekap@postech.ac.kr}.}
\and Sang Won Bae\thanks{Kyonggi University, South Korea.
Email: {swbae@kgu.ac.kr}.}
\and Otfried Cheong\thanks{KAIST, South Korea.
Email: {otfried@kaist.edu}.}
\and Joachim  Gudmundsson\thanks{University of Sydney and NICTA,
  Australia. Email: {joachim.gudmundsson@sydney.edu.au}.}
\and Takeshi Tokuyama\thanks{Tohoku University, Japan.
  Email: {tokuyama@dais.is.tohoku.ac.jp}.}
\and Antoine Vigneron\thanks{KAUST, Saudi Arabia.
  Email: {antoine.vigneron@kaust.edu.sa}.}}

\begin{document}
\maketitle

\begin{abstract}
  Given a set of line segments in the plane, not necessarily finite,
  what is a convex region of smallest area that contains a translate
  of each input segment?  This question can be seen as a
  generalization of Kakeya's problem of finding a convex region of
  smallest area such that a needle can be rotated through 360 degrees
  within this region.  We show that there is always an optimal region
  that is a triangle, and we give an optimal $\Theta(n \log n)$-time
  algorithm to compute such a triangle for a given set of $n$
  segments. We also show that, if the goal is to minimize the
  perimeter of the region instead of its area, then placing the
  segments with their midpoint at the origin and taking their convex
  hull results in an optimal solution. Finally, we show that for any
  compact convex figure $G$, the smallest enclosing disk of~$G$ is a
  smallest-perimeter region containing a translate of every rotated
  copy of~$G$.
\end{abstract}

\section{Introduction}

Let $\FAM$ be a family of objects in the plane.  A \emph{translation
  cover} for $\FAM$ is a set~$K$ such that any object in~$\FAM$ is
contained in a translate of $K$~\cite{Wetzel1973}.  We are interested
in determining a \emph{convex} translation cover for~$\FAM$ of
smallest possible area or perimeter.

Since the convex hull of a set of objects is the smallest convex
figure that contains them, this problem can be reformulated as
translating the objects in~$\FAM$ such that the perimeter or the area
of their convex hull is minimized.  When $\FAM$ consists of $n$
objects, we can fix one object and translate the remaining
$n-1$~objects. Therefore we can use a vector in $\Reals^{2(n-1)}$ to
represent the translations of $n-1$~objects.  Consider the functions
$\Reals^{2(n-1)} \to \Reals$ that take a vector in $\Reals^{2(n-1)}$
and return the perimeter and the area of the convex hull of the fixed
object and the translated copies of the $n-1$ other objects. Ahn and
Cheong~\cite{ahncheong2011} showed that for the perimeter case, this
function is convex.  They also showed that for the area case, the
function is convex if $n = 2$.  However, this is no longer true when
$n > 2$, as the following example shows.  Let $s_{1}$ be a vertical
segment of length one, and let $s_{2}$ and $s_{3}$ be copies of
$s_{1}$ rotated by~$60^{\circ}$ and~$120^{\circ}$. Then the area of
their convex hull is minimized when they form an equilateral triangle,
so there are two isolated local minima, as shown in
\figurename~\ref{fig:nonconvex}.  This explains why minimizing the
perimeter appears to be a much easier problem than minimizing the area
of a translation cover.
\begin{figure}[tb]
  \center
  \includegraphics[width=.6\textwidth]{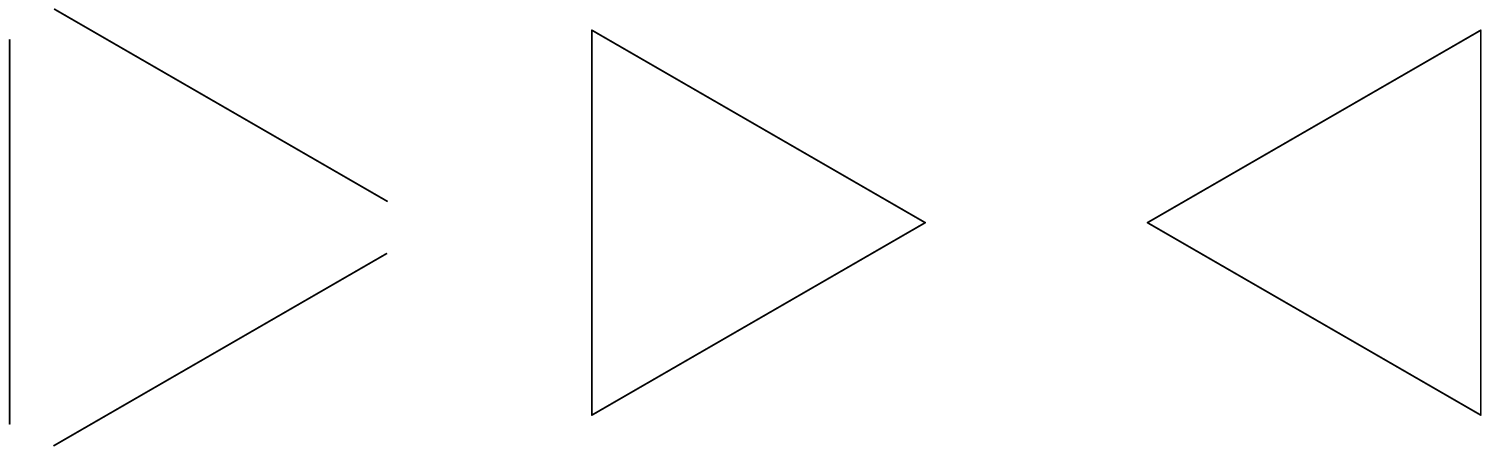}
  \caption{The area function $\omega \colon \Reals^{2(n-1)} \to
    \Reals$ of the convex hull of $n\geq 3$ segments is not
    necessarily convex.}
  \label{fig:nonconvex}
\end{figure}

As a special case of translation covers, we can consider the situation
where the family~$\FAM$ consists of copies of a given compact convex
figure~$G$, rotated by all angles in~$[0,2\pi)$.  In other words, we
are asking for a smallest possible convex set~$K$ such that $G$ can be
placed in~$K$ in every possible orientation.  We will call such a
translation cover a \emph{keyhole} for~$G$ (since a key can be turned
fully in a keyhole, it can certainly be placed in every possible
orientation).

A classical keyhole or translation cover problem is the Kakeya needle
problem.  It asks for a minimum area region in the plane, a so-called
\emph{Kakeya set}, in which a needle of length~$1$ can be rotated
through $360^{\circ}$ continuously, and return to its initial
position. (See Figure~\ref{fig:example}.)  This question was first
posed, for convex regions, by Soichi Kakeya in
1917~\cite{kakeya1917}. P\`al~\cite{pal1921} showed that the solution
of Kakeya's problem for convex sets is the equilateral triangle of
height one, having area $1/\sqrt 3$.  With our terminology, he
characterized the smallest-area keyhole for a line segment.
\begin{figure}[hbp]
  \centerline{\includegraphics[width=.25\textwidth]{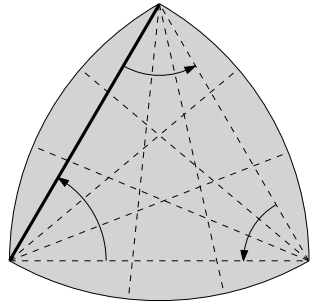}}
  \caption{Within a Kakeya set (shaded), a needle can be
	rotated through $360^{\circ}$.}
  \label{fig:example}
\end{figure}

For the general case, when the Kakeya set is not necessarily convex or
even simply connected, the answer was thought to be a deltoid with
area $\pi/8$. However, Besicovitch gave the surprising answer that one
could rotate a needle using an arbitrary small
area~\cite{besicovitch1920,besicovitch1928}.

\iftrue
Besicovitch's solution builds upon two basic
observations~\cite{tao01}. The first observation is that one can
translate any needle to any location using arbitrarily small area. The
idea is to slide the needle, rotate it, slide it back and then rotate it
back, as illustrated in Fig.~\ref{fig:besc_translation}(a).
\begin{figure}[hbp]
  \centerline{\includegraphics[width=.9\textwidth]{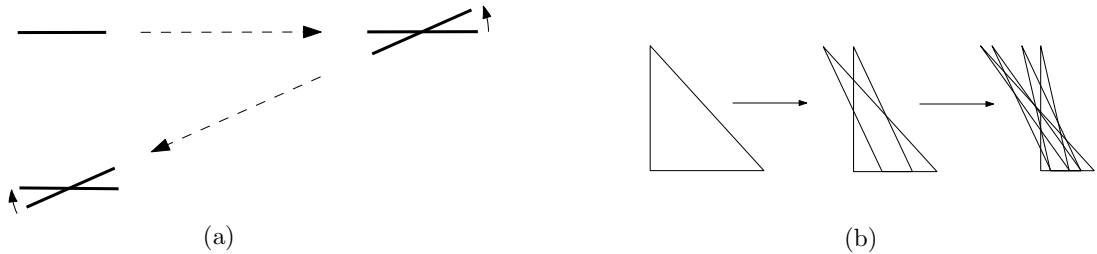}}
  \caption{(a) A needle can be translated to any location using
    arbitrarily small area. (b) There is an open subset of the plane
    of arbitrary small area which contain a unit line segment in every
    direction.}
  \label{fig:besc_translation}
\end{figure}
The area can be made arbitrarily small by sliding the needle over a
large distance. The second observation is that one can construct an
open subset of the plane of arbitrary small area, which contains a
unit line segment in every direction, as illustrated in
Fig.~\ref{fig:besc_translation}(b).  The original construction by
Besicovitch~\cite{besicovitch1920,besicovitch1928} has been simplified
by Perron~\cite{Perron1928}, Rademacher~\cite{Rademacher1962},
Schoenberg~\cite{Schoenberg1962a,Schoenberg1962b},
Besicovitch~\cite{besicovitch1963,besicovitch1964} and
Fisher~\cite{Fisher1973}.
\fi

Bezdek and Connelly~\cite{BezdekConnelly} surveyed results on
minimum-perimeter and minimum-area translation covers.  For the family
of closed curves of length at most one, they proved that
smallest-perimeter translation covers are exactly the convex sets of
constant width~$1/2$.  The corresponding problem for minimizing the
area, known as Wetzel's problem, is still open, with upper and lower
bounds known~\cite{BezdekConnelly,Wetzel1973}.  For the family of sets
of diameter at most one, Bezdek and
Connelly~\cite{Bezdek-connelly-1998} proved that the unique
minimum-perimeter translation cover is the circle of
radius~$1/\sqrt{3}$.  More precisely, they proved that this circle is
the unique smallest-perimeter keyhole for the equilateral triangle of
side length one.  By Jung's theorem~\cite{Jung1901}, this circle
contains any set of diameter one, and so the translation cover result
follows.

Recently, Kakeya-type problems have received considerable attention
due to their many applications. There are strong connections between
Kakeya-type problems and problems in number
theory~\cite{Bourgain2000}, geometric combinatorics~\cite{Wolff1999},
arithmetic combinatorics~\cite{Laba2008}, oscillatory integrals, and
the analysis of dispersive and wave equations~\cite{tao01}.

In this paper, we first generalize P\'al's result~\cite{pal1921} in
the following way: For \emph{any} family $\FAM$ of line segments in
the plane, there is a triangle that is a minimum-area translation
cover for~$\FAM$.
\begin{theorem}
  \label{thm:triangle}
  Let $\FAM$ be a set of line segments in the plane, and let $P$ be a
  convex translation cover for $\FAM$.  Then there is a translation
  cover $T$ for $\FAM$ which is a triangle, and such that the area
  of~$T$ is less than or equal to the area of~$P$.
\end{theorem}
With this characterization in hand, we can efficiently compute a
smallest area translation cover for a given family of~$n$ line
segments.  Our algorithm runs in time~$O(n \log n)$, which we prove to
be optimal in the algebraic computation tree model.  It is based on the
problem of finding a smallest-area affine-regular hexagon containing a
given centrally symmetric polygon, a problem that is interesting in
its own right.  As far as we know, except for some trivial cases such
as $n$ disks or $n$ axis-aligned squares, previously known algorithms
for finding smallest-area translation covers have a running time
exponential in $n$, the number of input 
objects~\cite{ahncheong2011,vigneron10}.

As observed above, minimizing the \emph{perimeter} of a translation
cover is much easier.  Let $\FAM$ be a family of centrally symmetric
convex figures.  We prove that if we translate each figure such that
its center of symmetry is the origin, then the convex hull of their
union is a smallest-perimeter translation cover for~$\FAM$.

This immediately implies that a circle with diameter~$1$ is a
smallest-perimeter keyhole for the unit-length segment.  For
figures~$G$ that are not centrally symmetric, this argument no longer
works.  We generalize the result by Bezdek and
Connelly~\cite{Bezdek-connelly-1998} mentioned above and prove the
following theorem (Bezdek and Connelly's result is the special case
where $G$~is an equilateral triangle):
\begin{theorem}
  \label{thm:keyholes}
  Let $G$ be a compact convex set in the plane, and let $\mathcal G$
  be the family of all the rotated copies of $G$ by angles in
  $[0,2\pi)$.  Then the smallest enclosing disk of $G$ is a
    smallest-perimeter translation cover for~$\mathcal G$.
\end{theorem}

\section{Preliminaries}

An \emph{oval} is a compact convex figure in the plane. For an
oval~$P$, let $w_{P}: [0,\pi] \rightarrow \R$ denote the width
function of~$P$.  The value $w_{P}(\Th)$ is the length of the
projection of $P$ on a line with slope~$\Th$ (that is, a line that
makes angle~$\Th$ with the $x$-axis).  Let $|P|$ denote the area
of~$P$.

For two ovals~$P$ and~$Q$, we write $w_{P} \geq w_Q$ or $w_Q \leq w_P$
to mean pointwise domination, that is for every $\Th \in [0,\pi)$ we 
have $w_{P}(\Th) \geq w_{Q}(\Th)$.  We also write $w_P = w_Q$ if and
only if both $w_P \leq w_Q$ and $w_Q \leq w_P$ hold.

\begin{figure}[tbh]
  \center
  \includegraphics[width=.7\textwidth]{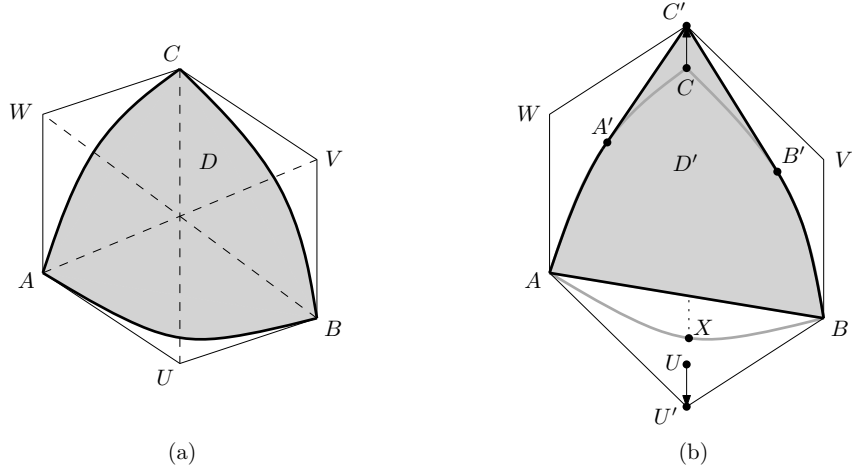}
  \caption{(a) A trigonal disk $D$ that is contained in the centrally
    symmetric hexagon $AUBVCW$ and contains the triangle $ABC$. (b)
    The hexagon $AU'BVC'W$ is centrally symmetric and contains
    $D'$. Since $D'$ contains the triangle $ABC'$, it is also a
    trigonal disk.}
  \label{fig:trigonal}
\end{figure}

The \emph{Minkowski symmetrization} of an oval~$P$ is the oval~$\Pb =
\frac 12 (P - P) = \{\frac 12 (x - y) \mid x, y \in P\}$.  It is well
known and easy to show that~$\Pb$ is centrally symmetric around the
origin, and that~$w_{\Pb} = w_P$.

An oval $D$ is a \emph{trigonal disk} if there is a centrally
symmetric hexagon $\mathit{AUBVCW}$ such that $D$ contains the
triangle $ABC$ and is contained in the hexagon~$\mathit{AUBVCW}$, as
illustrated in \figurename~\ref{fig:trigonal}(a).  Trigonal disks were
called ``relative Reuleaux triangles'' by Ohmann~\cite{ohmann1952} and
Chakerian~\cite{chakerian1966}, the term ``trigonal disk'' being due
to Fejes T\'oth~\cite{fejestoth1983} who used it in the context of
packings by convex disks. A trigonal disk has three ``main'' vertices
and three arcs connecting these main vertices. For example, the
trigonal disk $D$ in \figurename~\ref{fig:trigonal}(a) consists of
three vertices $A, B$, and $C$, and three arcs connecting them.

Ohmann~\cite{ohmann1952} and Chakerian~\cite{chakerian1966} studied
sets with a given fixed width function, and obtained the following result
(see for instance Theorem~$3'$ in~\cite{chakerian1966} for a proof):
\begin{fact}
  \label{thm:chakerian}
  Given an oval~$P$, there is a trigonal disk~$D$ with $|D| \leq |P|$
  such that $w_D = w_P$.
\end{fact}

\section{Minimum area for a family of segments}

In this section we will prove Theorem~\ref{thm:triangle}. The proof
contains two parts. First we prove that for every oval $P$ there
exists a triangle $T$ with $|T| \leq |P|$ and $w_T \geq w_P$
(Theorem~\ref{thm:min-width}). The second part is to prove that for an
oval $P$ and a closed segment $s$, if $w_s\leq w_P$ then $P$ contains
a translated copy of~$s$ (Lemma~\ref{lem:seg}).
\begin{theorem}
  \label{thm:min-width}
  Given an oval~$P$, there exists a triangle~$T$ with $|T| \leq |P|$
  and $w_T \geq w_P$.
\end{theorem}
\begin{proof}
  Let $\Ds$ be the set of trigonal disks $D$ such that we have $|D|
  \leq |P|$ and $w_D = w_P$.  The set $\Ds$ is nonempty by
  Fact~\ref{thm:chakerian}.  Consider three arcs connecting the main
  vertices of a trigonal disk in $\Ds$. Each arc can be straight, or
  not. We choose a trigonal disk $D \in \Ds$ with a maximum
  number of straight arcs.  We show that~$D$ is a
  triangle.

  Let $\mathit{AUBVCW}$ be the hexagon from the definition of the
  trigonal disk~$D$, and assume for a contradiction that~$D$ is not a
  triangle, that is, there is at least one non-straight arc among the
  three arcs connecting~$A, B$, and~$C$. See
  \figurename~\ref{fig:trigonal}(a). Without loss of generality, we
  assume that the arc connecting~$A$ and~$B$ is not straight.

  Let the sides $AW$ and $BV$ be vertical, with $C$ above the
  line~$AB$.  Let $X$ be the point of~$D$ below~$AB$ with the largest
  \emph{vertical} distance~$d$ from the line~$AB$.  Let $C'$ be the
  point vertically above~$C$ at distance~$d$ from~$C$.  Let $D'$ be
  the convex hull of the part of~$D$ above the line $AB$ and the
  point~$C'$.  It is not difficult to see that $D'$ is also a trigonal
  disk: Let $U'$ be the point vertically below~$U$ at distance~$d$
  from~$U$.  Then the hexagon $AU'BVC'W$ is centrally symmetric and
  contains $D'$. Clearly $D'$ contains the triangle $ABC'$. See
  \figurename~\ref{fig:trigonal}(b).

  We show next that~$|D'| \leq |D|$.  The area of $D'\setminus D$ is
  bounded by the area of the two triangles~$A'C'C$ and~$B'C'C$,
  where~$A'$ and~$B'$ are points on~$D$ such that $A'C'$ and~$B'C'$
  are tangent to~$D$.  This area is equal to $d/2$~times the
  horizontal distance between~$A'$ and~$B'$.  But the horizontal
  distance between $A'$ and~$B'$ is at most the horizontal distance
  between~$A$ and~$B$, so the area of~$D'\setminus D$ is bounded by
  the area of the triangle~$AXB$, and we have~$|D'| \leq |D|$.

  We also need to argue that $w_{D'} \geq w_D$.  Consider a minimal
  strip containing~$D$. If this strip does not touch~$D$ from below
  between $A$ and $B$, then the corresponding strip for~$D'$ is at
  least as wide. Otherwise, it touches~$D$ from below in a point~$Y$
  between~$A$ and~$B$, and touches from above in~$C$, as~$C$ is the
  only antipodal point of~$D$ for~$Y$.  A strip with the same
  direction will be determined either by $A$ and~$C'$, or by $B$ and
  $C'$, and in both cases its width is not less than the width of the
  original strip.

  Since $w_{D'} \geq w_D \geq w_P$ and $|D'|\leq |D|\leq|P|$ the
  trigonal disk $D'$ must be a member of $\Ds$.  However, $D'$ has at
  least one straight arc more than~$D$, contradicting our choice
  of~$D$.  It follows that our assumption that~$D$ is not a triangle
  must be false.
\end{proof}

This finishes the first part. We need the following lemma, which shows
that whether or not an oval~$P$ contains a translated copy of a given
segment~$s$ can be determined by looking at the width functions of~$P$
and~$s$ alone:
\begin{lemma}
  \label{lem:seg}
  Let $s$ be a segment in the plane, and let $P$ be an oval such that
  $w_s \leq w_P$.  Then $P$ contains a translated copy of~$s$.
\end{lemma}
\begin{proof}
  Without loss of generality, let $s$ be a horizontal segment.  Let
  $pq$ be a horizontal segment of maximal length contained
  in~$P$. Then $P$ has a pair of parallel tangents $\ell_1$
  and~$\ell_2$ through~$p$ and~$q$.  By the assumption, the distance
  between $\ell_1$ and~$\ell_2$ must be large enough to place~$s$ in
  between the two lines.  But this implies that the segment~$pq$ is at
  least as long as~$s$, and $s$ can be placed on the segment~$pq$
  in~$P$.
\end{proof}

To prove Theorem~\ref{thm:triangle}, let~$P$ be an oval of minimum
area that contains a translated copy of every~$s \in \FAM$.  By
Theorem~\ref{thm:min-width} there is a triangle~$T$ such that $|T|
\leq |P|$ and $w_T \geq w_P$.  Let $s \in \FAM$.  Since there is a
translated copy of~$s$ contained in~$P$, we must have $w_s \leq w_P
\leq w_T$. By Lemma~\ref{lem:seg} there is then a translated copy of
$s$ contained in~$T$.

\section{From triangles to hexagons}

We now turn to the computational problem: Given a family $\FAM$ of
line segments, find a smallest-area convex set that contains a
translated copy of every~$s \in \FAM$.

By Theorem~\ref{thm:triangle} we can choose the answer to be a
triangle. In this section we show that this problem is equivalent to
finding a smallest-area \emph{affine-regular hexagon} enclosing some
centrally symmetric convex figure.  An {affine-regular hexagon} is the
image of a regular hexagon under a non-singular affine transformation.
In this paper, we only consider affine-regular hexagons that are
centrally symmetric about the origin, so by abuse of terminology, we
will write \emph{affine-regular hexagon} for an affine-regular hexagon
that is centrally symmetric about the origin.

In the next section we will then show how to solve that problem, using
the tools of computational geometry.

The basic insight is that for centrally symmetric figures, comparing
width-functions is equivalent to inclusion:
\begin{lemma}
  \label{lem:width-is-inclusion}
  Let~$P$ and~$Q$ be ovals centrally symmetric about the origin.  Then
  $w_P \leq w_Q$ if and only if $P \subset Q$.
\end{lemma}
\begin{proof}
  One direction is trivial, so consider for a contradiction the case
  where $w_P \leq w_Q$ and $P \not\subset Q$.  Then there is a
  point~$p \in P \setminus Q$.  Since $Q$ is convex, there is a line
  $\ell$ that separates $p$ from~$Q$.  Since $P$ and $Q$ are centrally
  symmetric, this means that $Q$ is contained in the strip bounded by
  the lines~$\ell$ and~$-\ell$, while $P$ contains the points~$p$
  and~$-p$ lying outside this strip.  This implies that for the
  orientation~$\Th$ orthogonal to~$\ell$ we have $w_P(\Th) >
  w_Q(\Th)$, a contradiction.
\end{proof}

Recall that $\Pb$ denotes the Minkowski symmetrization of an oval $P$.
\begin{lemma}
  \label{lem:aff_hexagon}
  Let~$T$ be a non-degenerate triangle. Then $\Tb$ is an
  affine-regular hexagon, and $|\Tb| = \frac 32 |T|$.  Every
  affine-regular hexagon $H$ can be expressed in this form.
\end{lemma}
\begin{proof}
  Since every non-degenerate triangle is the affine image of an
  equilateral triangle, it suffices to observe this relationship for
  the equilateral triangle and the regular hexagon.
\end{proof}

Since $w_P = w_{\Pb}$, $w_T = w_{\Tb}$, and by
Lemmas~\ref{lem:width-is-inclusion} and~\ref{lem:aff_hexagon}, we
immediately have
\begin{lemma}
  \label{lem:p_smallest_hexagon}
  Given an oval~$P$, a triangle $T$ is a smallest-area triangle
  with~$w_T \geq w_P$ if and only if $\Tb$ is a smallest-area affine
  regular hexagon with $\Pb \subset \Tb$.
\end{lemma}

This leads us to the following algorithm.  In
Section~\ref{sec:lowerbound}, we will show that the time bound is
tight.
\begin{theorem} \label{thm:alg}
  Let $\FAM$ be a set of $n$ line segments in the plane.
  Then we can find a triangle $T$ in $O(n \log n)$ time which
  is a minimum-area convex translation cover for $\FAM$.
\end{theorem}
\begin{proof}
  Given a family~$\FAM$ of~$n$ line segments, place every $s\in \FAM$
  with its center at the origin.  Let~$P$ be the convex hull of these
  translated copies.  $P$ can be computed in $O(n \log n)$ time, and
  is a centrally symmetric convex polygon with at most $2n$~vertices.
  We then compute a smallest area affine-regular hexagon~$H$
  containing~$P$. In the next section we will show that this can be
  done in time~$O(n)$.  Finally, we return a triangle~$T$ with $\Tb =
  H$.  The correctness of the algorithm follows from $w_P(\Th) =
  \max_{s \in \FAM}w_s(\Th)$ and Lemma~\ref{lem:p_smallest_hexagon}.
\end{proof}

\section{Algorithm for computing the smallest enclosing affine-regular hexagon}
\label{sec:algorithm}

In this section we discuss the following problem: Given a convex
polygon~$P$, centrally symmetric about the origin, find a
smallest-area affine-regular hexagon~$H$ such that $P \subset H$.

Let us first sketch a simple quadratic-time algorithm: The
affine-regular hexagons centered at the origin are exactly the images
of a regular hexagon centered at the origin under a non-singular
linear transformation.  Instead of minimizing the hexagon, we can fix
a regular hexagon~$H$ with center at the origin, and find a linear
transformation~$\sigma$ such that~$\sigma P \subset H$ and such that
the determinant of~$\sigma$ is maximized.  The transformation $\sigma$
can be expressed as a $2\times 2$ matrix with coefficients~$a, b, c,
d$.  The condition $\sigma P \subset H$ can then be written as a set
of~$6n$ linear inequalities in the four unknowns~$a, b, c, d$.  We
want to find a feasible solution that maximizes the determinant~$ad -
bc$, a quadratic expression.  This can be done by computing the
4-dimensional polytope of feasible solutions, and considering every
facet of this polytope in turn.  We triangulate each facet, and solve
the maximization problem on each simplex of the triangulation.

In the following, we show that the problem can in fact be solved in
linear time.

For a set $S \subset \Reals^{2}$, let $S\m = -S$ denote the mirror
image with respect to the origin.  A \emph{strip} is the area bounded
by a line~$\ell$ and its mirror image~$\ell\m$.

An affine-regular hexagon~$H$ is the intersection of three
strips~$\strip_1$, $\strip_2$, and~$\strip_3$, as in
\figurename~\ref{fig:strips}, where the sides of~$H$ are supported
by~$\strip_1$, $\strip_2$, $\strip_3$, $\strip_1$, $\strip_2$, and
$\strip_3$ in counter-clockwise order.
\begin{figure}[htb]
  \centerline{\includegraphics{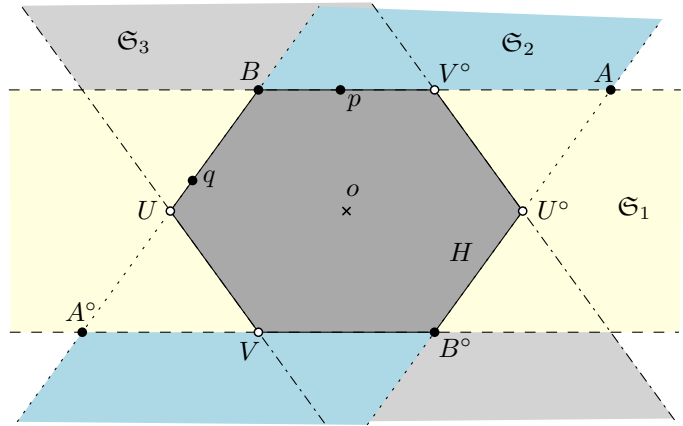}}
  \caption{The hexagon $H$ is defined by three strips.}
  \label{fig:strips}
\end{figure}
The intersection $\strip_1\cap\strip_2$ is a parallelogram~$Q = ABA\m
B\m$. Since $H$ is affine-regular, the sides supported by~$\strip_3$
must be parallel to and half the length of~$BB\m$, and so~$\strip_3$
is uniquely defined by~$\strip_1$ and~$\strip_2$: It supports the
sides~$UV$ and~$U\m V\m$ of~$H$, where $U$ is the midpoint of~$BA\m$
and $V$ is the midpoint of $A\m B\m$.  Note that $|H| = 3|Q|/4$.

It is easy to see that if $H$ is a minimum-area affine-regular hexagon
containing~$P$, then two of the three strips must be
touching~$P$. Without loss of generality, we can assume these to be
strips~$\strip_1$ and~$\strip_2$, so there is a vertex~$p$ of~$P$ on
the side $V\m B$, and a vertex~$q \in P$ on the side~$BU$.

For convenience of presentation, let us choose a coordinate system
where~$\strip_1$ is horizontal.  If we now rotate~$\strip_2$
counter-clockwise while remaining in contact with~$P$, then one side
rotates about the point~$q$, while the opposite side rotates
about~$q\m$, see \figurename~\ref{fig:rotate}.
\begin{figure}[htb]
  \centerline{\includegraphics{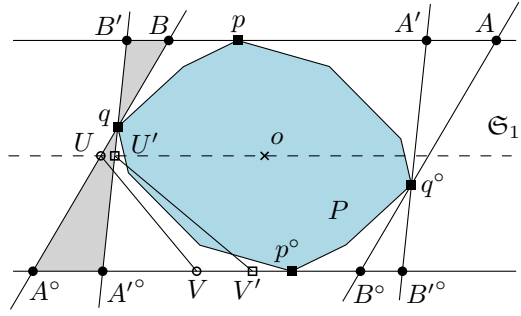}}
  \caption{Rotating strip~$\strip_2$ counter-clockwise.}
  \label{fig:rotate}
\end{figure}
The triangles~$qBB'$ and~$qA\m{A'}\m$ are similar, and since $q$ lies
above or on the $x$-axis, we have $|qA\m {A'}\m| \geq |qBB'|$. This
implies that the area of~$Q$ is nonincreasing
during this rotation.  Since $|H| = 3|Q|/4$, the area of~$H$ decreases
or remains constant as well.

Furthermore, the point~$U$ moves horizontally along the $x$-axis to
the right.  The point~$A\m$ moves horizontally to the right with at
least twice the speed of point~$U$.  As $V$ is the midpoint
of $A\m$ and $B\m$, this implies that $V$ moves
horizontally to the right with at least the speed of~$U$, and so the
line~$UV$ is rotating counter-clockwise.

It follows that while strip~$\strip_2$ rotates counter-clockwise, the
part of~$H$ lying below the $x$-axis and to the left of the
line~$pp\m$ is strictly shrinking.  It follows that there is a unique
orientation of~$\strip_2$ where the side~$UV$ touches~$P$, and the
area of~$Q$ is minimized.

Let us say that a polygon~$S$ is \emph{circumscribed to} another
polygon~$R$ if and only if $R\subset S$ and every side of~$S$
contains a point of~$R$.  Then we have shown
\begin{lemma}
  \label{lem:all-six}
  There is a minimum-area affine-regular hexagon~$H$ such
  that~$H$ is circumscribed to~$P$.
\end{lemma}
In fact, we have shown that for every~$\strip_1$ there is a
unique~$\strip_2$ such that $H$ is circumscribed to~$P$.  We have
\begin{lemma} \label{lem:s1-s2}
  When $\strip_1$ rotates counter-clockwise, then the
  corresponding~$\strip_2$ also rotates counter-clockwise.
\end{lemma}
\begin{proof}
  Consider a configuration where $H$ is circumscribed to~$P$, and
  rotate $\strip_1$ slightly around~$p$ in counter-clockwise
  direction, keeping~$\strip_2$ fixed.  Then $B$ and~$A\m$ move
  downwards along the line~$BA\m$, see \figurename~\ref{fig:rotate2}.
  \begin{figure}[htb]
    \centerline{\includegraphics{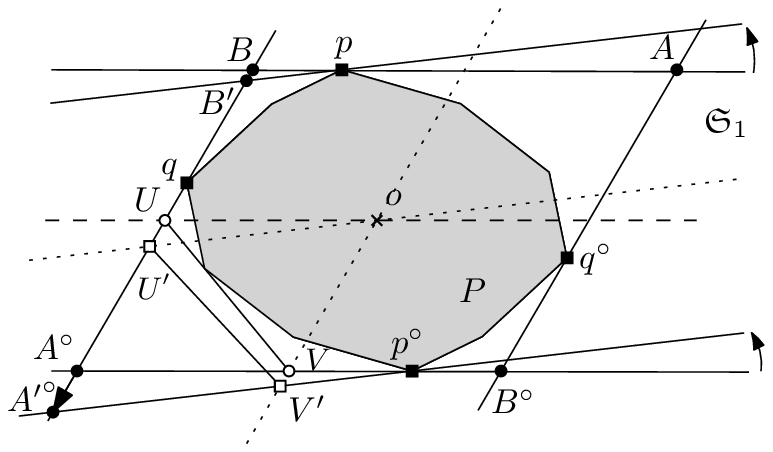}}
    \caption{Rotating~$\strip_1$ counter-clockwise.}
    \label{fig:rotate2}
  \end{figure}

  The point~$V$ moves downwards along the line~$oV$, parallel
  to~$BA\m$. It follows that the new edge $U'V'$ now lies strictly
  outside the old hexagon~$H$, and so $U'V'$ cannot possibly touch or
  intersect~$P$.  By the arguments above, this implies that
  strip~$\strip_2$ now needs to rotate counter-clockwise as well to
  let $H$ be circumscribed to~$P$.

  Furthermore, similar to the arguments above, we observe that $A\m$
  moves with speed at least twice the speed of~$V$.  Since $U$ is the
  midpoint of~$BA\m$, it moves with at least half the speed of~$A\m$,
  so $U$ moves with speed at least equal to the speed of~$V$. Since
  $U$ and $V$ move on parallel lines, it follows that the line $UV$ is
  rotating counter-clockwise during the rotation of~$\strip_1$.
\end{proof}

We can now show that we can in fact choose~$H$ such that one of its
sides contains an edge of~$P$:
\begin{lemma}
  \label{lem:hexagon}
  There exists a minimum-area affine-regular hexagon $H$ containing $P$
  such that a side of~$H$ contains an edge of~$P$. In addition, if
  no minimum-area affine-regular hexagon containing $P$ shares a
  vertex with $P$, then each such minimum-area affine-regular hexagon
  has a side containing a side of $P$.
\end{lemma}
\begin{proof}
  By Lemma~\ref{lem:all-six}, there exists a minimum-area affine-regular
  hexagon~$H$ such that every side of $H$ contains a point of~$P$.
  If a side of $H$ contains an edge of~$P$, then we are done.  In the
  following, we thus assume that every side of $H$ intersects $P$ in
  a single point.  Also, we assume the vertices of $H$ are $(1,0),
  (1,1), (0, 1)$ and their antipodal points $(-1, 0), (-1, -1), (0,
  -1)$.  This can be done by applying a nonsingular linear
  transformation, see \figurename~\ref{fig:hexagon}.
  \begin{figure}[h]
    \centerline{\includegraphics{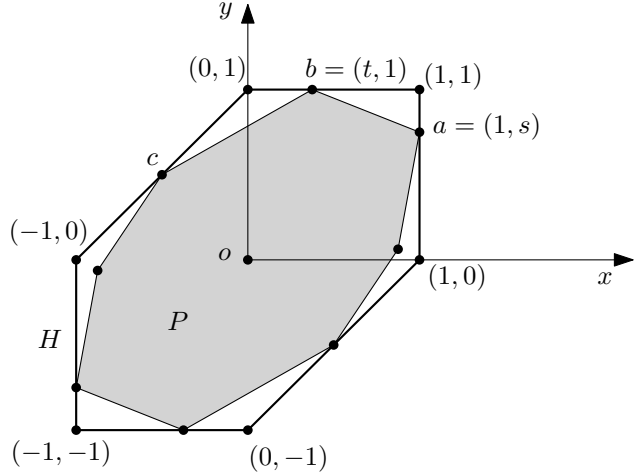}}
    \caption{The hexagon $H$, and the convex polygon $P$
      (shaded).}
    \label{fig:hexagon}
  \end{figure}

  First, we consider the case where no vertex of~$H$ coincides with a
  vertex of~$P$.  We claim that in this case, there exists a
  nonsingular linear transformation $\sigma$ such that $\sigma P
  \subset H$ and $|\sigma P| > |P|$ hold, implying that the inverse
  image $\sigma^{-1}H$ of $H$ also contains $P$ and its area
  $|\sigma^{-1} H|$ is strictly smaller than $|H|$, a contradiction.
  We denote by $a,b,c$ the three contact points as in
  \figurename~\ref{fig:hexagon}.  The point $c$ is a linear
  combination $c=\alpha a+\beta b$ of $a$ and $b$, so we have
  $c=(\alpha+\beta t,\alpha s+\beta)$. Since $c$ does not lie
  in the same quadrant as $a$ and $b$, nor the opposite quadrant,
  then $\alpha \beta \neq 0$ and $\alpha/\beta<0$.
  As the point $c$ lies on the
  line with equation $y=x+1$, we have
  $
  t=\frac{\alpha}{\beta}s+\frac{\beta-\alpha-1}{\beta}.
  $
  Then the area of the triangle $oab$ is given by
  \begin{equation}
    2|oab|=1-\frac{\alpha}{\beta}s^2-\frac{\beta-\alpha-1}{\beta} s.
    \label{eq:oab}
  \end{equation}
  Assume we apply a linear transformation $\sigma$ to $P$ such that
  each point in $a,b,c$ moves along the side of $H$ that currently contains
  it. Thus, $s$ changes, and $t$ changes in such a way that $c$
  remains on the same side of $H$. Then the area of $P$ is
  proportional to the area $|oab|$. As we observed that $\alpha/\beta < 0$,
  then the coefficient of $s^2$ in Equation~\eqref{eq:oab} is positive,
  so the area $|oab|$ cannot be at a local maximum, a contradiction.

  Consider now the case where at least one of the contact points $a,
  b, c$ lies at a vertex of~$H$.  Since each side of~$H$ has a
  single-point intersection with~$P$, two of $a, b, c$ are identical.
  Using a suitable linear transformation, we can assume $b = c$ in
  \figurename~\ref{fig:hexagon}.  Any linear transformation~$\sigma$
  that keeps $b = c$ fixed and moves~$a$ along the vertical side
  of~$H$ keeps areas unchanged, and so $|\sigma P| = |P|$ and thus
  $|\sigma^{-1} H| = |H|$.  Hence, there exists a linear
  transformation $\sigma'$ such that $\sigma' P \subset H$, $|\sigma'
  P| = |P|$, and $\sigma' P$ has one more contact point with the sides
  of $H$.
\end{proof}

We can therefore assume that the minimum-area affine-regular hexagon
is defined by two strips $\strip_1$ and~$\strip_2$, where $\strip_1$
supports an edge of~$P$, and $\strip_2$ is the unique strip such that
the resulting hexagon is circumscribed to~$P$.  
We now give a
linear-time algorithm to enumerate these hexagons, over all edges
of~$P$. 

\begin{theorem}
  \label{thm:min-circumscribed-hexagon}
  Given a centrally-symmetric convex $2n$-gon~$P$, a smallest-area
  affine-regular hexagon enclosing~$P$ can be found in time~$O(n)$.
\end{theorem}
\begin{proof}
  We use a rotating calipers~\cite{Toussaint83} type algorithm.  It
  maintains an edge~$e$ of~$P$ defining~$\strip_1$, a second
  strip~$\strip_2$ and the vertex~$q$ of~$P$ where $\strip_2$
  touches~$P$, and a vertex~$r$ of~$P$.  Let $H = BUVB\m U\m V\m$ be
  the hexagon defined by $\strip_1$ and~$\strip_2$, as in
  \figurename~\ref{fig:strips}.

  The algorithm proceeds by rotating~$\strip_2$ around~$q$ as in
  \figurename~\ref{fig:rotate}, and maintains the invariant that~$P$
  has a supporting line in~$r$ that is parallel to~$UV$.

  We initialize~$e$ to an arbitrary edge of~$P$.  Let $e$ be
  horizontal for ease of presentation, with~$P$ below~$e$, let $q$ be
  the left endpoint of~$e$, and let $r$ be the leftmost vertex of~$P$.
  In the initial configuration, $\strip_2$ is obtained from~$\strip_1$
  by a counter-clockwise rotation around~$q$ by an infinitely small
  amount.

  We then rotate $\strip_2$ counter-clockwise, until one of the
  following events occurs:
  \begin{denseitems}
  \item If $r$ no longer supports a tangent to~$P$ parallel to~$UV$,
    replace $r$ by the counter-clockwise next vertex of~$P$, and
    continue rotating~$\strip_2$.
  \item If $\strip_2$ supports an edge of~$P$, then replace $q$ by the
    counter-clockwise next vertex of~$P$, and continue
    rotating~$\strip_2$.
  \item If $UV$ touches~$r$, then we have found the unique $\strip_2$
    such that $H$ is circumscribed to~$P$.  We compute its area and
    update a running minimum.  Then replace $e$ by the
    counter-clockwise next edge of~$P$.  As long as $r$ does not
    support a tangent to~$P$ parallel to~$UV$, we replace $r$ by the
    counter-clockwise next vertex of~$P$.  Then continue
    rotating~$\strip_2$.
  \end{denseitems}
  The algorithm ends when $n$~edges have been considered.  Its
  running time is clearly linear.
\end{proof}

\section{Lower bound for computing a translation cover}
\label{sec:lowerbound}
In this section, we prove an $\Omega(n\log n)$ lower bound for the
problem of computing a minimum-area translation cover for a set
of $n$ line segments.
We first need the following result on regular
$6n$-gons (see Figure~\ref{fig:regular}(a).):

\begin{lemma}\label{lem:regular}
Let $R$ denote a regular $6n$-gon centered at the origin, for some
integer $n \geq 1$. Then any minimum-area affine-regular hexagon
enclosing $R$ is a regular hexagon such that every edge of
this hexagon contains an edge of $R$.
\end{lemma}
  \begin{figure}[h]
    \centerline{\includegraphics[width=.9\textwidth]{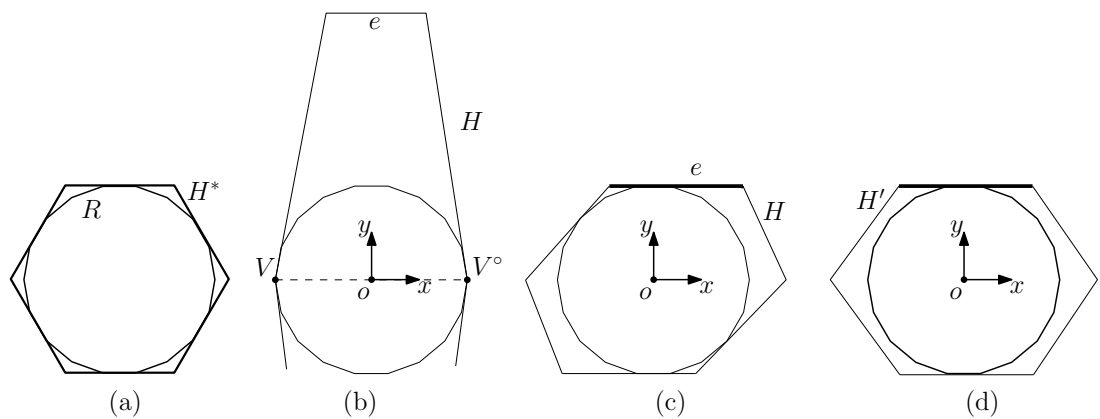}}
    \caption{Proof of Lemma~\ref{lem:regular}.
	(a) An optimal enclosing hexagon $H^*$ and the regular 18-gon $R$.
	(b) When $H$ and $R$ share two vertices, the area of $H$ is
	larger than the area of $H^*$.
	(c) An affine-regular enclosing hexagon $H$.
      (d) The hexagon $H'$.}
    \label{fig:regular}
  \end{figure}
\begin{proof}
The statement is trivial for $n=1$, so assume $n \geq 2$.
Let $H^*$ denote a regular hexagon enclosing $R$, and such that
each side of $H^*$ contains a side of $R$. Let $H$ denote another
smallest affine-regular hexagon enclosing $R$. We will argue that
$H$ is also a regular hexagon whose sides contain sides of $R$.

We first  rule out the case where $H$ shares a vertex with $R$.
For sake of contradiction, assume that $H$ shares two
opposite vertices $V$ and $V^\circ$ of $R$. Without loss of
generality, we assume that $V,V^\circ$ are on the $x$-axis.
The edges $e,e^\circ$ of $H$ that are not adjacent to $V,V^\circ$
are parallel to $VV^\circ$ and have half the length of~$\overline{VV^\circ}$.
In addition, the edges of $H$ that are adjacent to $V$ and $V^\circ$
make an angle at most $\pi/12$ with the $y$-axis.
(See Figure~\ref{fig:regular}(b).)
Then a direct calculation shows that $|H|>|H^*|$, a contradiction.

Thus, by Lemma~\ref{lem:hexagon}, we know that an edge $e$ of $H$
contains an edge of $R$. Without loss of generality we assume that
this edge is parallel to the $x$-axis. (See
Figure~\ref{fig:regular}(c).)  For sake of contradiction, assume that
$e$ is not symmetric with respect to the $y$-axis.  Consider the
hexagon $H'$ that is obtained from $H$ by a horizontal shear
transformation that moves~$e$ and the opposite edge parallel to the
$x$-axis, until they are centered at the $y$-axis.  Then $H'$ (see
Figure~\ref{fig:regular}(d)) is an affine-regular hexagon
containing~$R$ that is symmetric with respect to the $y$-axis and that
only touches~$R$ along its top and bottom edges. This implies that
$H'$ strictly contains a regular hexagon~$H^*$ enclosing~$R$, and
hence $|H|=|H'|>|H^*|$, a contradiction.
 
Therefore, $e$ is symmetric
with respect to the $y$-axis, and thus $H$ is symmetric with respect
to the $y$-axis. Only one such affine-regular hexagon is circumscribed
to~$R$, so $H=H^*$.
\end{proof}

We are now able to prove our lower bound.
\begin{theorem}\label{th:lowerbound}
In the algebraic computation tree model,
and in the worst case, it takes $\Omega(n \log n)$ time
to compute a minimum-area translation cover for a
family $\FAM$ of $n$ line segments in the plane.
\end{theorem}
\begin{proof}
For an interval $I \subset \R$, we denote by $C_{I}$ the arc of the unit circle
corresponding with polar angles in the interval $I$, that is
$C_{I}=\{(\cos \theta,\sin \theta) \mid  \theta \in I\}$. As
$C_{I}$ is the intersection of a circle and a cone,
a node of an algebraic computation tree can decide whether
a point lies in $C_{I}$.

We use a reduction from the following problem.
The input is a set of points $p_1,\dots,p_{n} \in C_{[0,\pi/3)}$.
The goal is to decide whether there exists an integer $0 \leq k < n$
such that $C_{(k\pi/3n,(k+1)\pi/3n)}$ is empty, that is, this
arc does not contain any
point $p_i$. It follows from Ben-Or's bound~\cite{Ben83} that any algebraic
computation tree that decides this problem has depth $\Omega(n\log n)$.
(The set of negative instances has at least $n!$ 
connected component: To each permutation $\sigma$ of $1,\dots,n$ ,  we associate 
a negative instance
where each $p_i$ lies in the $\sigma_i$'s arc. In order to move continuously
from one of these configuration to another, we must have a crossing $p_i=p_j$, which
implies that one interval is empty by the pigeonhole principle, and thus
the instance is positive.)

Our construction is as follows. Consider the (fixed) regular $6n$-gon $R$,
whose vertices are $r_k=(\cos (k\pi/3n),\sin (k\pi/3n))$ for
$k=1,\dots,6n$. Let $P$ denote the convex $12n$-gon whose vertices are
the vertices of $R$ and all the rotated copies of the
points $p_1,\dots,p_n$
by angles $0,\pi/3,\dots,5\pi/3$ around the origin.

If there is an integer $k=0,\dots,n-1$ such that $C_{(k\pi/3n,(k+1)\pi/3n)}$
is empty, then by Lemma~\ref{lem:hexagon},
the regular hexagon containing $R$ whose edges contain the edge $r_kr_{k+1}$
and its rotated copies by angles $0,\pi/3,\dots,5\pi/3$ is
a minimum area affine-regular hexagon containing $P$.

If on the other hand, for every integer $k \in \{0,\dots,n-1\}$ the arc
$C_{(k\pi/3n,(k+1)\pi/3n)}$ is nonempty, then by Lemma~\ref{lem:hexagon},
any minimum-area affine hexagon containing $R$ is a regular hexagon
whose edges contain edges of $R$, and thus it cannot contain $P$.

So we have proved that, when some arc $C_{(k\pi/3n,(k+1)\pi/3n)}$ is
empty, then a minimum-area hexagon containing $P$ has area $|H^*|$, where
$H^*$ is a minimum-area hexagon containing $R$. Otherwise, if all these
arcs are non-empty, then the minimum area is larger than $|H^*|$.

Thus, if we could compute in $o(n \log n)$ time
a minimum-area convex translation cover for the diagonals of $P$,
then by Lemma~\ref{lem:p_smallest_hexagon} we would also get
in $o(n \log n)$ time the area of a smallest enclosing affine-regular
hexagon containing $P$, and then we would be able to decide
in $o(n \log n)$ time whether there exists an empty
arc $C_{(k\pi/3n,(k+1)\pi/3n)}$, a contradiction.
\end{proof}

\section{Minimizing the perimeter}\label{sec:perimeter}

If we wish to minimize the perimeter instead of the area, the problem
becomes much easier: it suffices to translate all segments so that
their midpoints are at the origin, and take the convex hull of the
translated segments.  This follows from the following more general
result.
\begin{theorem} \label{thm:segment-perimeter}
  Let $\mathcal{C}$ be a family of centrally symmetric convex figures.
  Under translations, the perimeter of the convex hull of their union
  is minimized when the centers coincide.
\end{theorem}
\begin{proof}
  By the Cauchy-Crofton formula~\cite{doCarmo}, the perimeter is the
  integral of the width of the projection over all directions.  We
  argue that the width is minimized when the centers coincide, for all
  directions simultaneously, implying the claim.

  Assume the objects are placed with their center at the origin. Let
  $p$ be a leftmost point of the convex hull.  It belongs to one of
  the objects $C \in \mathcal{C}$.  By symmetry, the mirror image of
  $p$ is then a rightmost point of the convex hull.  But this implies
  that the horizontal width of the convex hull is equal to the width
  of $C$, and therefore as small as possible.
\end{proof}

When the figures are not symmetric, our proof of
Theorem~\ref{thm:segment-perimeter} breaks down. However, we are able
to solve the problem for a family consisting of all the rotated copies
of a given oval. (Remember that an oval is a compact convex set.)  The
following theorem was already stated in the introduction.
\newtheorem*{keyholetheorem}{Theorem~\ref{thm:keyholes}}
\begin{keyholetheorem}
  Let $G$ be an oval, and let $\mathcal G$
  be the family of all the rotated copies of $G$ by angles in $[0,2\pi)$.
  Then the smallest enclosing disk of $G$ is a smallest-perimeter 
  translation cover for $\mathcal G$.
\end{keyholetheorem}
\begin{proof}
  We observe first that, if $G$ is a segment, then by
  Theorem~\ref{thm:segment-perimeter}, the smallest enclosing disk
  of~$G$ is a smallest-perimeter translation cover for~$\mathcal G$.

  Consider next the case where $G$ is an acute triangle.  Choose a
  coordinate system with origin at the center of the circumcircle
  of~$G$, and such that the circumcircle has radius one.  We wish to
  prove that any translation cover for $\mathcal G$ must have
  perimeter at least~$2\pi$, implying that the circumcircle is
  optimal.

  We borrow an idea of Bezdek and
  Connelly~\cite{Bezdek-connelly-1998}. Let $v_{1}$, $v_{2}$, $v_{3}$
  be the three vertices of~$G$.  By our assumptions, the origin lies
  in the interior of their convex hull, and the three vectors have
  length one.  The origin can be expressed as a convex combination $0
  =\sum_{i=1}^{3} \alpha_{i}v_{i}$ with $\alpha_{i} \geq 0$ and
  $\sum_{i=1}^{3}\alpha_{i} = 1$.  Let~$\delta_{i}$, for $i = 1, 2,
  3$, be the angle formed by~$v_{i}$ and the positive $x$-axis.

  Let $K$ be a translation cover for~$\mathcal G$ and let $h$ be the
  support function~\cite{Schneider} of $K$.  That is,
  $h(u)=\sup\{\dop{x}{u} \mid x \in K\}$ for any unit vector~$u$.
  We denote by $u_\theta=(\cos \theta,\sin \theta)$ the unit vector
  making angle~$\theta$ with the positive $x$-axis, so that 
  $v_i = u_{\delta_{i}}$.
  
  The length $\lambda$ of the perimeter of~$K$ is equal to the
  integral over the support function~\cite{strang}
  \[
  \lambda = \int_{0}^{2\pi} h(u_\theta) d\theta.
  \]
  Since $\theta \mapsto h(u_\theta)$ is a periodic function
  with period~$2\pi$, we have
 \[
  \lambda = \int_{0}^{2\pi} h(u_\theta) d\theta
  = \int_{\delta_{i}}^{2\pi+\delta_{i}} h(u_\theta) d\theta
  = \int_{0}^{2\pi} h(u_{\theta + \delta_{i}}) d\theta.
  \]
  It follows that
  \begin{equation}
    \label{eq:alpha}
  \lambda  = \sum_{i=1}^{3}\alpha_{i} \lambda
  = \sum_{i=1}^{3} \alpha_{i}\int_{0}^{2\pi} h(u_{\theta + \delta_{i}}) d\theta 
  = \int_{0}^{2\pi}\big( 
  \sum_{i=1}^{3} \alpha_{i}h(u_{\theta+\delta_{i}}) \big)d\theta.
  \end{equation}
  Consider now a fixed orientation~$\theta$.  The translation cover~$K$ must
  contain a rotated copy $G(\theta)$ of $G$ such that, for some translation
  vector~$c(\theta)$, the vertices of $G(\theta)$ are the points
  $v_{i}(\theta)=c(\theta)+u_{\theta+\delta_i}$ for $i=1,2,3$.

  Since $v_{i}(\theta)$ lies in $K$, the value of the support
  function~$h(u_{\theta+\delta_i})$ is lower bounded by
  \begin{equation}
    \label{eq:lower-bound}
    h(u_{\theta+\delta_i}) \geq \dop{v_i(\theta)}{u_{\theta+\delta_i}}=
    \dop{c(\theta)+u_{\theta+\delta_i}}{u_{\theta+\delta_i}}
    = \dop{c(\theta)}{u_{\theta+\delta_i}}+1
  \end{equation}
  and thus
  \[
  \sum_{i=1}^{3}\alpha_{i}h(u_{\theta+\delta_{i}})
  \geq
  1+\dop{c(\theta)}{\sum_{i=1}^{3} \alpha_{i}u_{\theta+\delta_i}} = 
  1 + \dop{c(\theta)}{0} = 1.
  \]
  Plugging this into Eq.~(\ref{eq:alpha}) gives $\lambda \geq 2\pi$.

  Consider finally the general case where $G$ is an arbitrary compact
  convex figure, 
  and let~$D$ be the smallest enclosing disk of~$G$.  Either $D$ touches $G$ in
  two points
  that form a diameter of~$D$, or $D$ touches $G$ in three points that
  form an acute triangle.  In both cases, our previous results imply
  that $D$ is a smallest-perimeter translation cover for either the segment
  or the triangle, and therefore for~$G$.
\end{proof}
The minimum enclosing circle is not always the unique minimum-perimeter
keyhole: For instance, when $G$ is a unit line segment, then any 
set of constant width is a solution. In the theorem below, we show that
when $G$ is an acute triangle, then its circumcircle is the unique
solution. This generalizes directly to any figure $G$ that touches its
circumcircle at 3 points.
\begin{theorem}
  If $G$ is an acute triangle, then its smallest enclosing disk is the
  unique smallest-perimeter translation cover for the family of all
  rotated copies of $G$.
\end{theorem}
\begin{proof}
  We use the same notations as in the proof of Theorem~\ref{thm:keyholes}:
  $K$ is a smallest-perimeter translation cover for $\mathcal G$.
  For any $\theta$, it contains a copy $G(\theta)$ of $G$ rotated by
  angle $\theta$. 
  The vertices of $G(\theta)$ are the points
  $v_i(\theta)=c(\theta)+u_{\theta+\delta_i}$, for $i=1,2,3$.

  We will prove that all the triangles $G(\theta)$ have the same
  circumcircle. Our strategy is to show that the function $\theta
  \mapsto c(\theta)$ is differentiable and its derivative is
  $0$. Without loss of generality, we only prove that $c'(0)=0$, and
  we assume that $c(0)=0$.

  For sake of contradiction, assume that $c$ is not differentiable at
  $0$, or it is differentiable at $0$ and its derivative is
  nonzero. This means that we do not have $\lim_{\theta \rightarrow 0}
  c(\theta)/\theta=0$. Hence, there exists an $\eps>0$ such that for
  any integer $n$, there exists $\theta_n \in (-1/n,0) \cup (0,1/n)$
  with~$\|c(\theta_n)/\theta_n\|>\eps$.  This implies
  $c(\theta_{n}) \neq 0$, and so $c(\theta_n)/\|c(\theta_n)\|$ is a
  sequence of unit vectors. Since the set of unit vectors is compact,
  there is a subsequence~$(\theta_{n_k})$ such that
  $c(\theta_{n_k})/\|c(\theta_{n_k})\|$ converges to a unit
  vector~$c_0$.  We denote this subsequence again as~$(\theta_n)$.
  
  Since $u_{\delta_1}, u_{\delta_2}, u_{\delta_3}$ span~$\R^{2}$,
  there exists $i\in \{1,2,3\}$ such that $\dop{c_0}{u_{\delta_i}}>0$. So
  \[\lim_{n \rightarrow \infty}\frac{1}{\|c(\theta_{n})\|}
  \dop{c(\theta_{n})}{u_{\delta_i}} = 
  \dop{c_0}{u_{\delta_i}}>0.\]
  As $\|c(\theta_{n})\|>\eps \|\theta_{n}\|$ for all~$n$,
  this implies that for $n$ large enough,
  \[\dop{c(\theta_{n})}{u_{\delta_i}}>
  \frac{\eps \|\theta_{n}\|}{2}\dop{c_0}{u_{\delta_i}}, 
  \]
  hence
  \begin{align*}
  \dop{v_i(\theta_{n})}{u_{\delta_i}} &= 
  \dop{c(\theta_{n})+u_{\theta_{n}+\delta_i}}{u_{\delta_i}}\\ 
  &> \frac{\eps \|\theta_{n}\|}{2}\dop{c_0}{u_{\delta_i}}+
  \cos(\theta_{n})\\
  & = 1+\frac{\eps \|\theta_{n}\|}{2}\dop{c_0}{u_{\delta_i}}-
  \frac{\theta_{n}^2}{2}+o(\theta_{n}^3).
  \end{align*}
  Thus, for large enough~$n$, we have
  $\dop{v_i(\theta_{n})}{u_{\delta_i}} > 1$.  Since $v_{i}(\theta_{n})
  \in K$ for all~$n$, this implies $h(u_{\delta_i}) > 1$.  But since
  $c(0) = 0$, this means $h(u_{\delta_i}) > 1 +
  \dop{c(0)}{u_{\delta_{i}}}$, and so
  Inequality~(\ref{eq:lower-bound}) in the proof of
  Theorem~\ref{thm:keyholes} is not tight.  Since the support function
  $h$ is continuous~\cite{Schneider}, this implies that $\lambda>2
  \pi$, a contradiction.
  \end{proof}

\section{Conclusions}

In practice, it is an important question to find the smallest convex
container into which a family of ovals can be translated.  For the
perimeter, this is answered by the previous lemma for centrally
symmetric ovals.  For general ovals, it is still not difficult, as the
perimeter of the convex hull is a convex function under
translations~\cite{ahncheong2011}.  This means that the problem can be
solved in practice by numerical methods.

For minimizing the area, the problem appears much harder, as there can
be multiple local minima.  The following lemma solves a very special
case.
\begin{lemma} \label{lem:16}
  Let $\mathcal{R}$ be a family of axis-parallel rectangles.  The area
  of their convex hull is minimized if their bottom left corners
  coincide (or equivalently if their centers coincide).
\end{lemma}
\begin{proof}
  Let $C$ be the convex hull of some placement of the rectangles.  For
  any $x$, let $\ell(x)$ be the length of the intersection of the
  vertical line at coordinate $x$ with $C$.  The function $x \mapsto
  \ell(x)$ is concave (by the Brunn-Minkowski theorem in two
  dimensions).  For any $z \geq 0$, we define $w(z)$ to be the length
  of the interval of all $x$ where $\ell(x) \geq z$.

  We observe that the area of $C$ is equal to $\int \ell(x) d x$,
  which is again equal to $\int_{0}^{\infty}w(z) d z$. We will now
  argue that $w(z)$ is minimized for every~$z$ when the bottom left
  corners of the rectangles coincide, implying the claim.

  To see this, consider the placement with coinciding bottom left
  corners at the origin, and the line $y=z$.  It intersects the convex
  hull at $x=0$ and at some convex hull edge defined by two rectangles
  $R_1$ and $R_2$.  $w(z)$ is equal to the length of this
  intersection.  It remains to observe that for any placement of $R_1$
  and $R_2$, the convex hull of these two rectangle already enforces
  this value of~$w(z)$.
\end{proof}

\section*{Acknowledgments}

We thank Helmut Alt, Tetsuo Asano, Jinhee Chun,
Dong Hyun Kim, Mira Lee, Yoshio Okamoto,
J\'anos Pach, 
G\"unter Rote,
and Micha Sharir for helpful discussions.

\bibliographystyle{abbrv}
\bibliography{kakeya}
\end{document}